\documentclass[journal,twoside,web]{ieeecolor}
\usepackage{graphicx,keyval,trig,url} %,biblatex}
\usepackage[T1]{fontenc}     %LaTeX Font Warning: Font shape `OMS/cmtt/m/n' undefined 
\usepackage[utf8]{inputenc}	%apostrophes appear normally
\usepackage{graphicx}

\usepackage{amsthm}
\usepackage{amssymb}
\usepackage{amsmath}
\usepackage{mathtools}
\usepackage[mathscr]{eucal}
\usepackage{lipsum}
\usepackage{cuted} % environment changes the formatting from two-column to one-column to better accommodate very long equations
\usepackage{blindtext}
\usepackage{color}
\usepackage{dsfont}

\usepackage{epstopdf}
\usepackage{bm}   %turn the letter into bold in the math mode

\usepackage{color}
\usepackage{xcolor}

\newtheorem{rem}{Remark}
\usepackage{multirow}
\usepackage{longtable}

\usepackage{mathtools}
\usepackage[colorinlistoftodos]{todonotes}

\DeclareMathOperator*{\argmax}{arg\,max}

\newtheorem{theorem}{Theorem}
\newtheorem{prop}{Proposition}
\newtheorem{assumption}{Assumption}
\newtheorem{corollary}{Corollary}
\newtheorem{problem}{Problem}
\usepackage{generic}
\usepackage{cite}
\usepackage{amsmath,amssymb,amsfonts}
\usepackage{algorithmic}
\usepackage{graphicx}
\usepackage{textcomp}
\def\BibTeX{{\rm B\kern-.05em{\sc i\kern-.025em b}\kern-.08em
    T\kern-.1667em\lower.7ex\hbox{E}\kern-.125emX}}
\begin{document}
\title{The Immersion and Invariance Wind Speed Estimator Revisited and New Results}
\author{
Yichao Liu, Atindriyo Kusumo Pamososuryo, Riccardo M.G. Ferrari and Jan-Willem van Wingerden
\thanks{This research was supported in part by the Netherlands Organization for Scientific Research (NWO) via a VIDI grant 17512, and in part by the European Union via a Marie Sklodowska-Curie Action (Project EDOWE, grant 835901).}
\thanks{
Yichao Liu, Atindriyo Kusumo Pamososuryo, Riccardo M.G. Ferrari and Jan-Willem van Wingerden are with the Delft Center for Systems and Control, Delft University of Technology, Mekelweg 2, 2628 CD Delft, the Netherlands (e-mail: {\{Y.Liu-17, A.K.Pamososuryo, R.Ferrari, J.W.vanWingerden\}@tudelft.nl}).}}
\pagestyle{empty} 
\maketitle
\thispagestyle{empty} 

\begin{abstract}
The Immersion and Invariance (I\&I) wind speed estimator is a powerful and widely-used technique to estimate the rotor effective wind speed on horizontal axis wind turbines.
Anyway, its global convergence proof is rather cumbersome, which hinders the extension of the method and proof to time-delayed and/or uncertain systems. 
In this letter, we illustrate that the circle criterion can be used as an alternative method to prove the global convergence of the I\&I estimator. This also opens up the inclusion of time-delays and uncertainties. 
First, we demonstrate that the I\&I wind speed estimator is equivalent to a torque balance estimator with a proportional correction term.
As the nonlinearity in the estimator is sector bounded,
the well-known circle criterion is applied to the estimator to guarantee its global convergence for time-delayed systems.
By looking at the theoretical framework from this new perspective, this letter further proposes the addition of an integrator to the correction term to improve the estimator performance.
Case studies show that the proposed estimator with an additional integral correction term is effective at wind speed estimation. 
Furthermore, its global convergence can be guaranteed by the circle criterion for time-delayed systems.

%\vspace{-0.2cm}
%% not exceeding 200 words %%
%% not exceeding six pages %%

%Wind speed information is usually required in modern wind turbine controllers for load reduction and power regulation.
%However, the wind speed measured on the turbine's nacelle is not an ideal representative of the Rotor Effective Wind Speed (REWS). 
%This single measurement point does not take into consideration the time-varying wind states over the entire rotor disk.
%In this context, effective wind speed estimators using turbine measurements have been proposed to tackle this problem. 
%Inspired by the immersion and invariance estimator, a novel yet simple algorithm called integral estimator is introduced in this letter.
%The power coefficient is computed to reflect the nonlinearity between the generator torque and the wind speed.
%Then an integral estimator, based on a circle criterion design, is developed to estimate the REWS.
%Its global convergence is guaranteed by quantifying the optimal adaptation gain for the time delayed estimator. 
%The effectiveness of the proposed estimator is illustrated in high-fidelity simulations of a wind turbine system.
%Case studies show that the proposed integral estimator is effective at estimating REWS in time-varying wind speed cases, while its global convergence can be guaranteed for different stability conditions with time delays.

\end{abstract}


\section{Introduction}
\label{sec:introduction}
\IEEEPARstart{W}IND energy has received increasingly considerable attention in the international energy markets in recent years. 
More than 60\,GW new wind power was installed in 2019, which demonstrates a 19\,\% growth compared to 2018 for the global wind industry~\cite{GWEC_2020}. 
In this context, the sizes of wind turbines are increased, which results in a rising demand for optimization of wind turbine controllers in the world.

In designing a wind turbine controller or wind farm controller, knowledge of wind speed over the rotor is widely used to improve the control performance, \emph{e.g.}, gain scheduling and feedback techniques~\cite{Ostergaard_2007,DOEKEMEIJER2020719}. 
Unfortunately, the detailed information on the effective wind speed over the entire rotor disk is still limited~\cite{Liu2021}. 
The measured wind speed from an anemometer deployed on the turbine's nacelle or spinner is not precise, as the device can only provide pointwise information.
The wind conditions, however, demonstrate high spatial variability over the rotor disk of a single wind turbine.
% Lidars or other advanced remote sensing techniques~\cite{Simley2013,Schlipf_2013} might be viable for measuring the spatio-temporal wind fields; nonetheless, they are still costly and hence not yet ready for deployment at an industrial scale ~\cite{Bottasso_2018}.
% This, to a large extent, restricts the applications of advanced control techniques in industrial practice. 

To address this issue, a number of rotor effective wind speed estimators~\cite{Soltani2013} have been proposed in the past years, among which the torque balance estimator class~\cite{Ostergaard_2007, Ma_1995, Ortega_2013} appears as one of the simplest and widely-used wind speed estimation solutions.
The basic idea is that the measured generator power or torque signals, together with the measured rotor speed, can be utilized to estimate the aerodynamic power or torque based on the turbine's power coefficient, and thereby the effective wind speeds. 
%Other techniques, such as Kalman filter~\cite{Song_2017}, local wind inflow estimator~\cite{Bottasso_2018, Schreiber_2020}, also showed promising results in wind speed estimation.
One of the most attractive torque balance methods is the so-called Immersion and Invariance (I\&I) technique, introduced in Ortega \textit{et al.}~\cite{Ortega_2013}.
The reasons for this is due to 1) its ease of tuning and guaranteed convergence; 2) its ability to take nonlinearity into account without linearization; and 3) being validated in the field~\cite{Soltani2013}.

The torque balance estimator can be seen as a type of Lur'e system~\cite{Khalil_2015} formed by a negative feedback interconnection of a linear estimator and a bounded nonlinearity on the turbine power coefficient.
In this context, a common restriction for such Lur'e type estimators is the lack of sufficient asymptotic stability conditions on the linear stable estimator such that the feedback interconnection is stable. 
Another challenge stems from the fact that these estimation approaches are usually implemented in discrete time, where time-delays may induce closed-loop instabilities~\cite{Lee_1981}.
% Therefore, it is important to analyse the stability conditions for such Lur'e type estimators and take into account the time-delays.
% To achieve the global convergence of torque balance estimators, Ortega \textit{et al.}~\cite{Ortega_2013} resorted to a so-called Immersion and Invariance (I\&I) technique~\cite{Liu_2009}, in which little to no attention was given to the presence of time-delays.
In the work of Ortega \textit{et al.}~\cite{Ortega_2013}, the proof of global convergence of such a torque balance estimator was provided.
However, the derivation of such a proof is rather cumbersome and, moreover, the presence of time-delays did not receive enough attention.
% For the latter, current researches are still not able to account for the stability conditions for time-delays.
% As the wind speed estimator is usually implemented in the digital systems with time-delays, the lack of available asymptotic stability conditions will limit their applications for real wind turbines.
This cannot be neglected as the wind speed estimator is usually implemented as a digital system, where the lack of available asymptotic stability conditions will limit its application to real wind turbines.

The main contribution of this letter is fourfold.
First, we show that the I\&I wind speed estimator is equivalent to a torque balance estimator with a proportional correction term.
%Considering that the estimation error can not be removed by the proportional form, an integral action is added to improve the estimation performance.
%The design introduced in this letter removes the I\&I framework for simplification, as the torque balance estimator boils down to a general proportional form.
Second, we present an alternative proof, with exactly the same assumptions and conditions, for the I\&I wind speed estimator based on the well-known circle criterion~\cite{Khalil_2015}. 
Third, we extend the proof by including time-delays in the proof.
Fourth, we will show that the performance of the estimator can be improved by adding an integrator to the correction term.
%That is, the nonlinearity is strictly monotonically increasing.
%This restriction guarantees that the static nonlinearity in the estimator will satisfy the sector condition of the circle criterion.
%As a result, a circle criterion estimator is formulated, in which the time-delays are explicitly taken into consideration.
%The adaptation gain of the estimator is computed to satisfy the circle criterion, which then drives the estimation error to zero, and thus, its global convergence is guaranteed.
The proposed estimator is finally verified in several case studies.

The remainder of this letter is structured as follows:
Section~\ref{sec:2} describes the fundamental wind turbine dynamics that will be incorporated in the wind speed estimator.
Section~\ref{sec:3} formulates the estimation problem, followed by its solutions in~\ref{sec:4}, alternative proof and extensions.
In Section~\ref{sec:5}, case studies are performed to illustrate the performance of the estimator and its global convergence properties for different stability conditions in the presence of time-delays. 
Finally, concluding remarks are presented in Section~\ref{sec:6}.

\section{Definition of a wind turbine model}\label{sec:2}
In this section the basic dynamics of a wind turbine is presented, which will be incorporated in the estimator. 
%The wind turbine model considered in this letter is based on a 5MW three-bladed variable-speed reference wind turbine developed by NREL~\cite{Jonkman_2009}.
The aerodynamic power captured from the wind is given by:
\begin{equation}
P_{w} = \frac{1}{2}\rho A U^3 C_{p}(\lambda)
\, ,
\label{eq:P_w}
\end{equation} 
with $\rho$, $A$, $U$, and $C_p(\cdot)$, as the air density, rotor swept area, rotor effective wind speed, and power coefficient, respectively.
Here, $C_p(\cdot)$ is a nonlinear function of the tip speed ratio defined as:
\begin{equation}
\lambda := \frac{\omega_rR}{U}
\, ,
\label{eq:lambda}
\end{equation}
where $\omega_r$ and $R$ denote the rotor speed and rotor radius, respectively.
Note that $C_p(\cdot)$ is typically presented also as a function of blade pitch angle in the literature, but, without loss of generality, we consider constant pitch angle throughout this work.
%For the blade pitch angle, the control signal is (?) formulated from the pitch controller to regulate generator power in Region \uppercase\expandafter{\romannumeral3}, where the turbine operates above the rated operation point. 
%Since we mainly focus in this letter on Region \uppercase\expandafter{\romannumeral2} where the turbine aims to extract the maximum wind power below the rated operation point, the dependence of the pitch angle is therefore omitted.
% The shape of $C_p(\cdot)$ is determined by the turbine design.
% It can be obtained from numerical simulations or experimental tests.
It is also worth mentioning that the shape of the $C_p(\cdot)$ curve relies on the design of the turbine and can be obtained either from numerical simulations or experimental data.

As an example, for the National Renewable Energy Laboratory (NREL) 5 MW wind turbine model~\cite{Jonkman_2009}, the $C_p(\cdot)$ curve covering the operating region of interest is illustrated in Fig.~\ref{Cp table}.
It is evident that for $\lambda \in [\lambda_{\text{min}}, \lambda_{\text{max}}]$, in which $\lambda_{\text{min}}>0$ and $\lambda_{\text{max}}>0$, $C_p$ satisfies:

\begin{figure}
\centering 
\includegraphics[width=0.8\columnwidth]{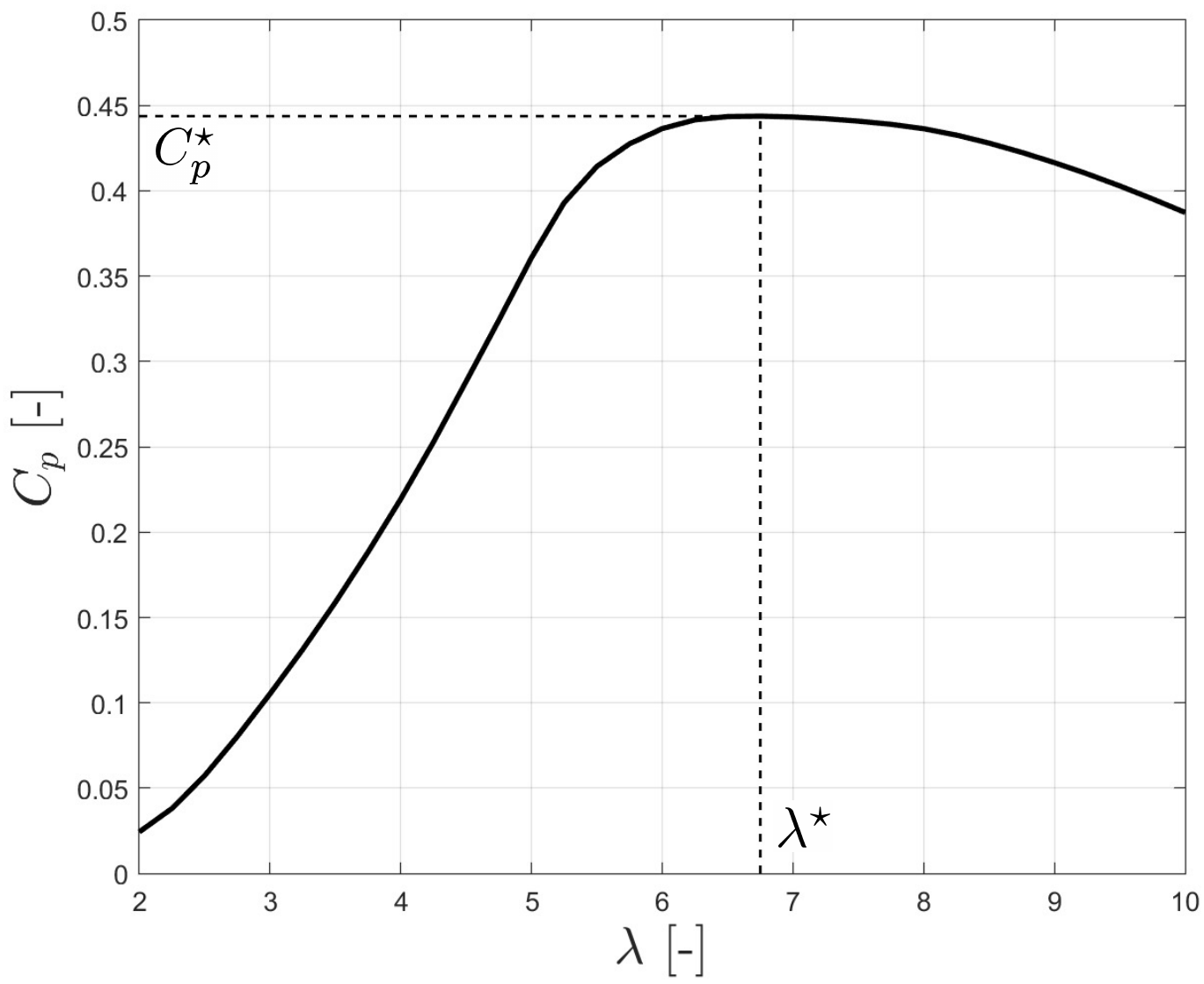}
\caption{Power coefficient for the 5MW wind turbine model~\cite{Jonkman_2009} covering the operating region of interest.}
\label{Cp table} %
\end{figure}

\begin{equation}
%0 < C_p^{\text{min}}\leq C_p(\lambda) \leq  C_p^\star
C_p(\lambda)>0
\, ,
\label{eq:C_p}
\end{equation}
%where $C_p^{\text{min}}$ and $C_p^\star$ denote the minimum and maximum values.
% and there exists a constant $\lambda^\star$, such that the wind turbine system reaches the maximum power generation point at $C_p^\star$.
and there exists a constant $\lambda^\star$ which corresponds to the maximum power generation point at $C_p^\star$.
In other words, $\lambda^\star$ is defined as:
\begin{equation}
\lambda^\star := \underset{\lambda}{\argmax{C_p(\lambda)}}
\, .
\label{eq:lambda_star}
\end{equation}
The dynamics of the wind turbine generator are given by: 
\begin{equation}
J\dot{\omega}_g = T_r/N - T_g
\, ,
\label{eq:Jwg}
\end{equation}
where $J$ is a known parameter describing the equivalent inertia at the generator shaft obtained from the relation $J=J_g+J_r/N^2$.
The symbols $J_g$ and $J_r$ are the inertia's of the generator and rotor while $N=\omega_g/\omega_r$ represents the gear ratio of the transmission with $\omega_g$ as the generator speed.
The symbols $T_g$ and $T_r$ denote the generator and aerodynamic torque, in which the latter can be derived from \eqref{eq:P_w} as follows:
\begin{equation}
T_r = \frac{P_w}{\omega_r} = \frac{\rho A}{2}\frac{U^3}{\omega_r} C_p(\lambda)
\, ,
\label{eq:T_r}
\end{equation}
and leads to the following definition of the nonlinearity:
\begin{equation}
\Phi(\omega_r, U) := \frac{T_r}{NJ} = \frac{\rho A}{2NJ}\frac{U^3}{\omega_r} C_p(\lambda)
\, .
\label{eq:Phi}
\end{equation}
% On the other hand, the generator power $P_g$ is derived from the relation:
% \begin{equation}
% P_g = \eta_g T_g \omega_g
% \, ,
% \label{eq:P_g}
% \end{equation}
% where $\eta_g \in [0,1]$ represents the generator efficiency, while the generator speed $\omega_g = \omega_r N$.

In this letter, the wind turbine satisfies the following assumption, under which wind speed estimation is possible:

\begin{assumption}\label{ass:rotor_lower_bound}
The rotor speed is positive and lower bounded. That is, for all $t\ge0$, there exists $\omega_r^{\text{min}}>0$, such that 
\begin{equation}
\omega_r(t) \ge  \omega_r^{\text{min}}
\, .
\label{eq:assumption1}
\end{equation}
\end{assumption}
%Considering the restricted $\lambda$ in \eqref{eq:C_p} and the definition in \eqref{eq:lambda}, this assumption means there also exists $U^{\text{min}}>0$, 
%\begin{equation}
%U^{\text{min}} := \frac{R\omega_r^{\text{min}}}{\lambda}
%\, ,
%\label{eq:U_min}
%\end{equation}
% which enables the wind speed estimation in this letter. 

For wind speed estimation, the well-known torque balance estimator usually employs a negative feedback interconnection between a linear estimator and a nonlinear function, namely the power coefficient. 
A general structure is presented in Fig.~\ref{general torque balance}, where $\hat{\omega}_r$ and $\epsilon$ are the estimated rotor speed and the error between estimated and measured rotor speed, respectively.
The linear estimator, which can be designed by using different estimation methods~\cite{Soltani2013}, aims to drive the error to zero, and thus obtain an optimal estimate $\hat{U}$ of $U$.

Ortega \textit{et al.} \cite{Ortega_2013} proposed an I\&I torque balance estimator, which is defined as
\begin{equation}
\begin{cases}
\label{eq:IandI}
\dot{\hat{U}}^I = \gamma \left[\frac{T_g}{NJ}-\frac{1}{N}\Phi\left(\omega_r, \hat{U}^I+\gamma\omega_r\right)\right] \\

\hat{U} = \hat{U}^I + \gamma\omega_r
\end{cases} \, ,
\end{equation}
where $\gamma>0$ is the proportional gain.

\begin{rem}
The gear ratio of the transmission $N$ is omitted in~\cite{Ortega_2013}. 
Without loss of generality, $N$ is considered in this letter.
\end{rem}

\begin{figure}
\centering 
\includegraphics[width=0.9\columnwidth]{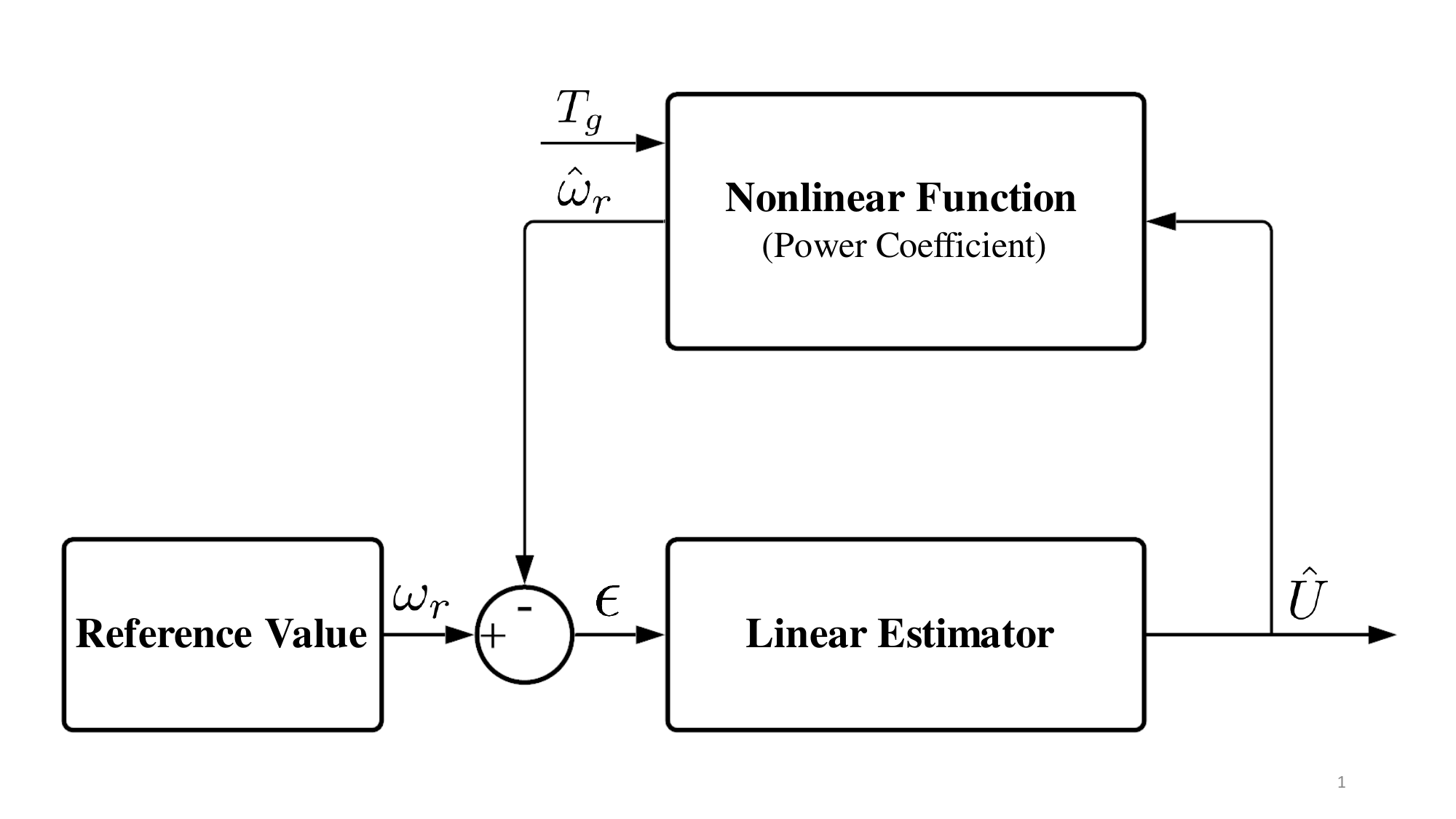}
\caption{Block diagram of a general torque balance estimator. It is formed by a negative feedback interconnection between the linear estimator and the power coefficient.}
\label{general torque balance} %
\end{figure}

\section{Problem statement}\label{sec:3}
In this section, the following assumptions are made to formulate the wind speed estimation problem. 

\begin{assumption}\label{ass:cp_smooth}
$C_p: [\lambda_{\text{min}}, \lambda_{\text{max}}]\to \mathbb{R}_{+}$ is a \textit{known} and smooth nonlinear function of class $\mathbb{C}^n$, that is its $0$--th through $n$--th derivatives are continuous, where $n$ is a non-negative integer. 
It satisfies \eqref{eq:C_p} and 

\begin{equation}
C_p'=
\begin{cases}
\label{eq:assumption2}
>0\,\,\,\,\,\, \text{for}\,\,\,\,\,\, \lambda \in [\lambda_{\text{min}}, \lambda^\star) \\
=0\,\,\,\,\,\, \text{for}\,\,\,\,\,\, \lambda = \lambda^\star \\
<0\,\,\,\,\,\,  \text{for}\,\,\,\,\,\, \lambda \in (\lambda^\star, \lambda_{\text{max}}]
\end{cases} \, ,
\end{equation}
in which $(\cdot)'$ represents differentiation.
\end{assumption}

\begin{assumption} \label{ass:constant_U}
The rotor effective wind speed $U$ is an \textit{unknown} positive constant.
\end{assumption}

\begin{rem}
Assumption~\ref{ass:constant_U} follows Ortega \textit{et al.}'s paper~\cite{Ortega_2013}. 
In practice, such a constant wind speed assumption can still approximately hold if the low amplitude oscillation noise in the measured signals can be filtered out, and the slowly-varying signals of the wind turbine, \emph{e.g.}, $\omega_r$, are successfully tracked~\cite{Fernando_2012}.
\end{rem}

\begin{assumption}\label{ass:measurable_Tg_wr}
$T_g$ and $\omega_r$ are the measured signals, and $\omega_r$:~$(0, \infty) \to (0, \infty)$ satisfies 
\begin{equation}
 \lim_{x\to\infty} \frac{\omega_r(bx)}{\omega_r(x)}=1
\, ,
\label{eq:assumption4}
\end{equation}
for all $b>0$.
\end{assumption}

The wind speed estimation problem addressed in this letter is therefore as follows:

%\begin{problem}\label{prob:alter_proof}
%For a given wind turbine system in \eqref{eq:P_w}, \eqref{eq:Jwg} %and its nonlinearity \eqref{eq:Phi}, find alternative global %convergence proof such that the I\&I estimator is capable of %providing an asymptotically consistent and online estimate %$\hat{U}$, and also opens up the extension to time-delayed and/or %uncertain systems.
%That is, an estimate of $U$ verifies

%\begin{equation}
%\lim_{t\to\infty} \hat{U}(t) = U
%\, .
%\label{eq:estimate}
%\end{equation}
%\end{problem}

\begin{problem}\label{prob:alter_proof}
For the given wind turbine system in \eqref{eq:P_w}, \eqref{eq:Jwg} and its nonlinearity \eqref{eq:Phi}, find an estimator that is capable of providing an asymptotically consistent estimate of $\hat{U}$ and $\hat{\omega}_g$ which is also robust to time-delays. 
%$alternative global convergence proof such that the I\&I estimator is capable of providing an asymptotically consistent and online estimate $\hat{U}$, and also opens up the extension to time-delayed and/or uncertain systems.
That is, an estimate of $U$ and $\omega_g $ such that:

\begin{equation}
\lim_{t\to\infty} \hat{U}(t) = U, \hspace{4mm} \lim_{t\to\infty} \hat{\omega}_g(t) = \omega_g
\, .
\label{eq:estimate}
\end{equation}
\end{problem}
%\begin{rem}
%The assumption of the slowly-varying $\omega_r$ in \textbf{Assumption 4} guarantees such external signals, i.e. $\omega_r$, $T_g$ entering the feedback loop of the torque balance estimator does not induce system instabilities. 
%\end{rem}

\section{Main results}\label{sec:4}
To derive the main results, the circle criterion~\cite{Khalil_2015, Karl_2020} is recalled here as follows. 
\begin{theorem}[Circle criterion]\label{thrm:circle_criterion}
Consider a negative feedback system consisting of a linear system $G(s)$ and a static sector-bounded nonlinearity $\Phi(x)$, satisfying

\begin{equation}
k_1 x\le \Phi(x)\le k_2 x
\, .
\label{eq:sector bound}
\end{equation}
The closed-loop interconnection is stable if the Nyquist curve of $G(s)$ does not enter a circle in the complex plane, with a radius of $\frac{k_2-k_1}{2k_1k_2}$ and center at $\frac{-k_2-k_1}{2k_1k_2}$, and the encirclement condition of the general Nyquist criterion is satisfied. 
\end{theorem}

\begin{proof}
The proof with respect to the circle criterion can be found in nonlinear system literature, \emph{e.g.},~\cite{Khalil_2015, Karl_2020}. 
Thus, it is omitted in this letter.
\end{proof}

\begin{prop}\label{prop:IandI_integral}
Consider the wind turbine system in \eqref{eq:P_w}, \eqref{eq:Jwg} and its nonlinearity \eqref{eq:Phi}, the following form is an equivalence of~\eqref{eq:IandI}
\begin{equation}
\begin{cases}
\label{eq:integral}
\dot{\hat{\omega}}_r=\frac{1}{N}\Phi\left(\omega_r,  \hat{U}\right)-\frac{T_g}{NJ}\\

\hat{U} = \gamma (\omega_r-\hat{\omega}_r)
\end{cases} \, .
\end{equation}
\end{prop}

\begin{proof}
The first equality can be readily derived from~\eqref{eq:Jwg} and~\eqref{eq:Phi}. 
Substituting both into~\eqref{eq:IandI} yields~\eqref{eq:integral}. 
\end{proof}

The theoretical framework~\eqref{eq:integral} is very useful as the I\&I estimator~\eqref{eq:IandI} is shown to be equivalent to a torque balance estimator with a proportional correction term.
This will significantly simplify the global convergence proof of the estimator.
\begin{corollary}
Let us consider the wind speed estimator in~\eqref{eq:integral} and define $\epsilon = \omega_r - \hat{\omega}_r$. 
Then the estimation error $\epsilon \rightarrow \epsilon_f \neq 0$ where $\epsilon_f$ is the estimator steady state error.
%Consider the theoretical framework for the wind speed estimator in~\eqref{eq:integral}.
%There exists a steady-state error $\epsilon$ in the estimation results.
\end{corollary}
From~\eqref{eq:integral}, we observe that if $\epsilon$ is zero, $\hat{U}$ would be zero as well, which causes an improper estimate of $U$ and a non-physical value for $\Phi$.
%The estimator steady state error originates from the fact that the proportional form~\eqref{eq:integral} is not able to converge $\epsilon_f$ to 0.
Based on this consideration, a Proportional Integral (PI) correction term is proposed in this letter, which results in an extension of~\eqref{eq:integral} as
\begin{equation}
\begin{cases}
\label{eq:proportional+integral}
\dot{\hat{\omega}}_r=\frac{1}{N}\Phi\left(\omega_r,  \hat{U}\right)-\frac{T_g}{NJ}\\

\epsilon = \omega_r-\hat{\omega}_r \\
\hat{U} = \gamma \epsilon + \beta \int_0^t \epsilon(\tau) d\tau
\end{cases} \, ,
\end{equation}
where $\beta$ is the integral gain, $t$ the present time, and $\tau$ the variable of integration.
%\end{corollary}
By adding the integrator, the extended formula~\eqref{eq:proportional+integral} would be able to equivalently make the error converge to zero, while provide estimates of not only wind speed~$\hat{U}$, but also the rotor speed~$\hat{\omega}_r$.

%\begin{corollary}
The overall structure of the proposed estimator with a PI correction term is presented in a block diagram with transfer functions, as shown in Fig.~\ref{integral block}.
The time-delay given by $e^{-sT}$ where $T$ is the sampling period, is explicitly considered.
Thus, the linear part of the proposed wind speed estimator with a PI correction term, denoted~$G(s)$ in Fig.~\ref{integral block}, can be rewritten as
\begin{equation}
G(s) = \frac{\gamma s+\beta}{s^2}e^{-sT}
\, .
\label{eq:G1}
\end{equation}
%where $e^{-sT}$ denotes the output time-delays.
%\end{corollary}
%\begin{rem}
%The slowly-varying signal $\omega_r$, as described in \textbf{Assumption 4}, are regarded as an external disturbance on the system, and thus not considered in the stability analysis.
%\end{rem}
In addition, Fig.~\ref{integral block} actually shows a normalized form of the proposed estimator, with the aid of the scaling factor $\alpha$. The value selected for $\alpha$ will be elaborated below.

\begin{figure}
\centering 
\includegraphics[width=1\columnwidth]{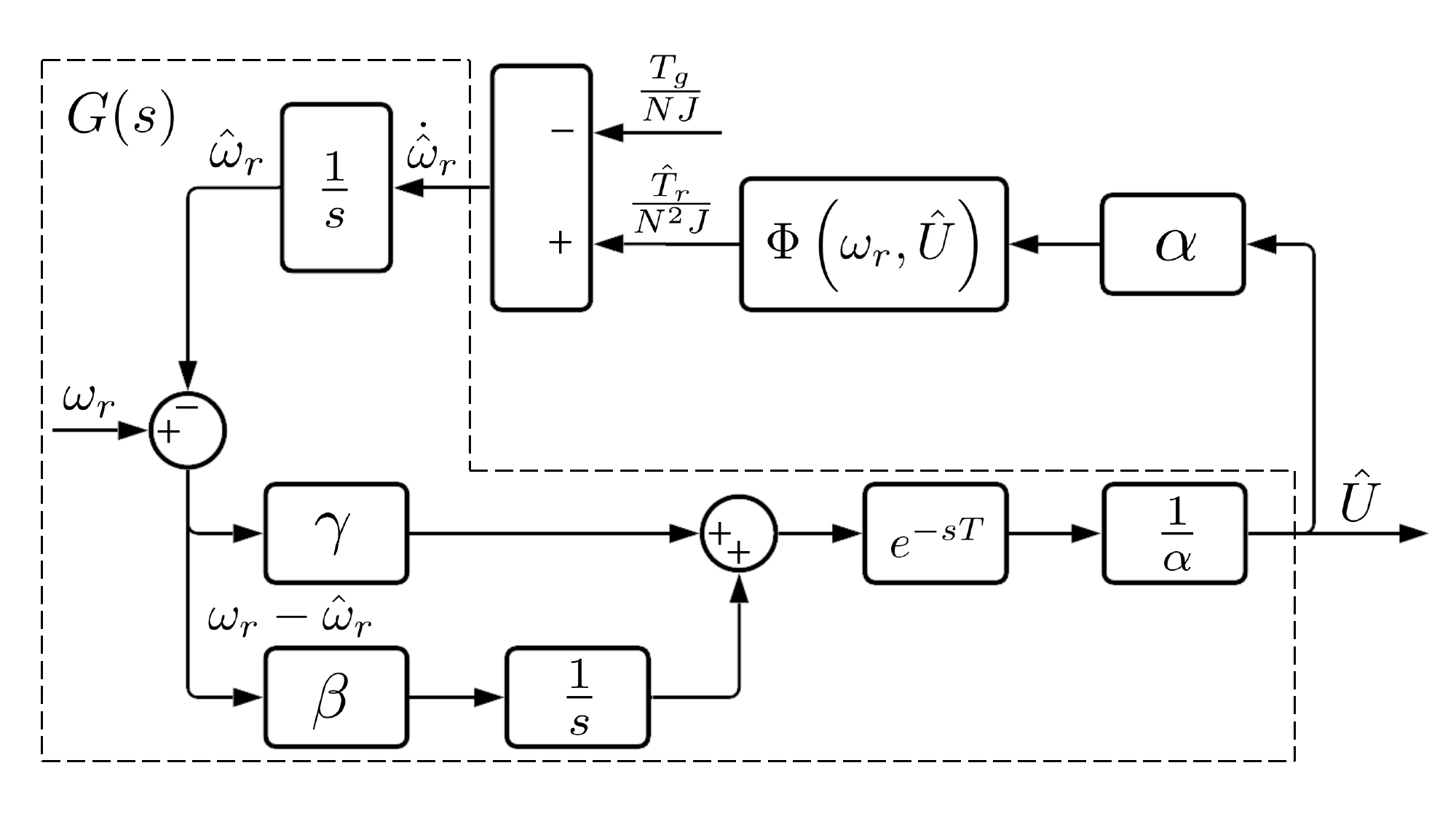}
\caption{Block diagram of the wind speed estimator in~\eqref{eq:proportional+integral} for the time-delayed system, where $G(s)$ is a minimum realization of the proposed wind speed estimator with a proportional integral correction term. $e^{-sT}$ denotes the time-delay. $\alpha$ represents the scaling factor to normalize $G(s)$.}
\label{integral block} %
\end{figure}

%\begin{prop} 
%Consider the nonlinearity $\Phi(\omega_r, U)$ defined in~\eqref{eq:Phi} with $C_p$ assumed in $Assumption 2$.
%There exists $\lambda_1$ and $\lambda_2$ with 
%\begin{equation}
%0<\lambda_1 \le \lambda_2 <\lambda^{\star}, 
%\, ,
%\label{eq:c0}
%\end{equation}
%such that 

%\end{prop} 
\begin{prop} 
The nonlinearity $\Phi(\omega_r, U)$ defined in~\eqref{eq:Phi} increases monotonically with respect to the input argument $U$, if either of the following conditions holds.

\textbf{Condition 1}. The power coefficient satisfies 
\begin{equation}
\frac{3}{\lambda}C_p(\lambda)>C_p^{'} (\lambda)
\, ,
\label{eq:c1}
\end{equation}
for $\lambda \in (\lambda_0, \lambda^\star]$, where $\lambda_0 \in (\lambda_\text{min}, \lambda^\star)$.

\textbf{Condition 2}. $\omega_r$ satisfies
\begin{equation}
\omega_r >\omega_r^{\lambda_0}
\, ,
\label{eq:c2}
\end{equation}
with the definition of $\omega_r^{\lambda_0}$
\begin{equation}
\omega_r^{\lambda_0} := \frac{\lambda_0 U}{R}
\, .
\label{eq:omega_lambda0}
\end{equation}

\end{prop}

\begin{proof}
The first order derivative of~\eqref{eq:Phi} is:

\begin{equation}
\Phi'(\omega_r,U)=\frac{\partial\Phi(\omega_r,U)}{\partial U} \\
= \frac{\rho ARU}{2NJ}(\frac{3}{\lambda}C_p(\lambda)-C_p'(\lambda))
\, ,
\label{eq:derivative}
\end{equation}
with the definition of the term $\kappa(\lambda)$

\begin{equation}
\kappa(\lambda) := \frac{3}{\lambda}C_p(\lambda)-C_p'(\lambda)
\, ,
\label{eq:kappa}
\end{equation}

The monotonicity of $\Phi(\omega_r, U)$ is based on the sign of the term $\kappa(\lambda)$, as $U$ is positive as assumed in Assumption~\ref{ass:constant_U}.
If $\kappa(\lambda)>0$, it can thus be proved that $\Phi(\omega_r, U)$ monotonically increases.

%It is clear that there exists one zero-crossing of $\kappa(\lambda)$ when $\lambda\approx 3$.

Let us prove there exists $\lambda_0<\lambda^\star$, such that for $\lambda \in (\lambda_0, \lambda_{\text{max}}]$, $\kappa(\lambda)>0$.
According to~\eqref{eq:assumption2}, it is concluded that $\kappa(\lambda)>0$ for all $\lambda \in [\lambda^\star, \lambda_{\text{max}}]$. 
This lies in the fact that $\frac{3}{\lambda}C_p(\lambda)>0$ for all $\lambda$ and $C_p'(\lambda)<0$ for $\lambda \ge\lambda^\star$.
On the other hand, considering the continuity and $C_p'(\lambda^\star)=0$, there exists $\lambda_0<\lambda^\star$ such that $\kappa(\lambda)>0$ for $\lambda \in (\lambda_0, \lambda^\star]$.
By combining both arguments together, the first condition is proved. 
In the condition of~\eqref{eq:c2}, $\lambda>\lambda_0$ is satisfied according to~\eqref{eq:lambda} and~\eqref{eq:omega_lambda0}, which completes the proof of the second condition.
%To conclude, there exists $\lambda_0 \in (\lambda_{\text{min}}, \lambda^{\star})$, such that for $\lambda \in (\lambda_0, \lambda_{\text{max}}]$, $\frac{3}{\lambda}C_p(\lambda)>C_p'(\lambda)$ is satisfied, and, hence, $\Phi(\omega_r, U)$ increases monotonically. 
\end{proof}
For the 5MW wind turbine model considered in this letter, $\kappa(\lambda)$, $\frac{3}{\lambda}C_p(\lambda)$ and $C_p'(\lambda)$ with respect to $\lambda$ are graphically illustrated in Fig.~\ref{kappa}, where $\lambda_\text{min}=2$ and $\lambda_\text{min}=10$, respectively.
\begin{figure}
\centering 
\includegraphics[width=0.8\columnwidth]{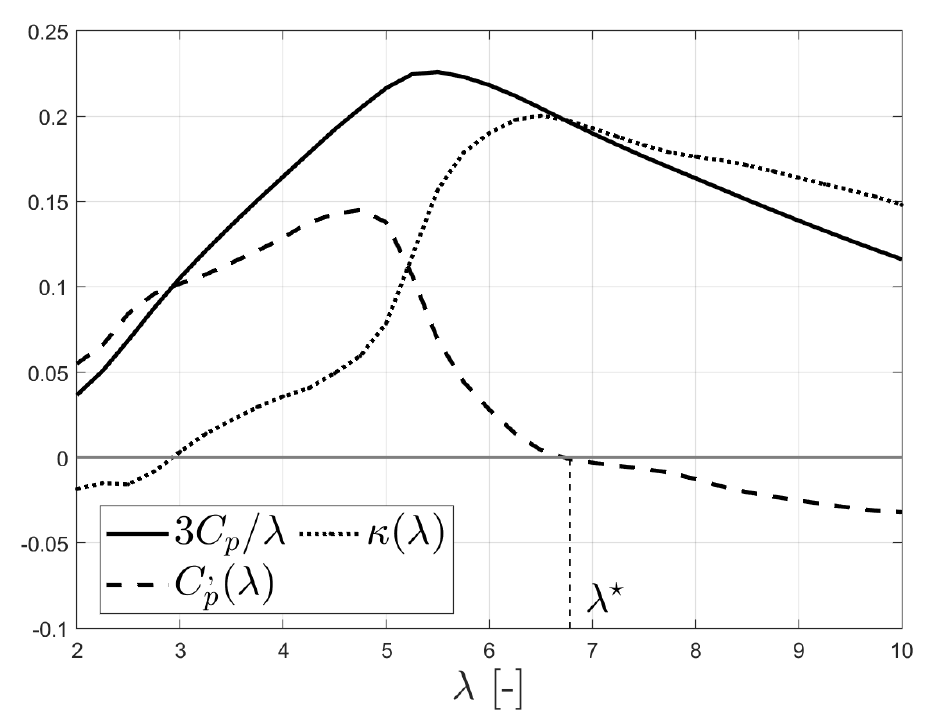}
\caption{Graphic illustration of $\kappa(\lambda)$, $\frac{3}{\lambda}C_p(\lambda)$ and $C_p'(\lambda)$ with respect to $\lambda$ for the 5MW wind turbine.}
\label{kappa} %
\end{figure}

\begin{prop} 
The nonlinearity $\Phi(\omega_r, U)$ defined in~\eqref{eq:Phi} with respect to $U$ satisfies the sector-bounded condition of the circle criterion in Theorem~1, that is, 
\begin{equation}
k_1 U \le \Phi(\omega_r, U) \le k_2 U
\, ,
\label{eq:setor bound}
\end{equation}
with the assumption that $0<k_1<k_2$.
\end{prop}

\begin{proof}
First, note the fact that
\begin{equation}
\frac{\rho A}{2NJ}>0 
\, ,
\end{equation}
with the conditions in~\eqref{eq:C_p} and~\eqref{eq:assumption1}, we have
\begin{equation}
k_1 \le \frac{\rho A}{2NJ}\frac{U^2}{\omega_r}C_p(\omega_r, U) \le k_2
\, .
\label{eq:setor bound proof step1}
\end{equation}
As $U$ satisfies Assumption~3, we conclude that
\begin{equation}
k_1 U \le \frac{\rho A}{2NJ}\frac{U^3}{\omega_r}C_p(\omega_r, U) \le k_2 U
\, ,
\label{eq:setor bound proof step2}
\end{equation}
with the definition of $\Phi(\omega_r, U)$ in~\eqref{eq:Phi}, completing the proof.
\end{proof}
Now we are in position to illustrate the main results.

\begin{theorem}
Consider the wind speed estimator with a PI correction term in~\eqref{eq:proportional+integral} and let us assume that either of Condition~1 and Condition~2 holds.
A sufficient condition for the global convergence of this wind speed estimator is that the Nyquist curve of the linear system with time-delays $G(s)$ does not enter the circle with a radius of $\frac{k_2-k_1}{2k_1k_2}$ and center at $\frac{-k_2-k_1}{2k_1k_2}$.
\end{theorem}

\begin{proof}
The proof of this theorem is readily derived from the circle criterion in Theorem~1 and Proposition~2 and 3.
\end{proof}

The distance criterion can be used to formulate the same sufficient stability conditions as in Theorem~2.

\begin{corollary}

Let us consider the wind speed estimator with a PI correction in~\eqref{eq:proportional+integral} and let us assume that either of Condition~1 and Condition~2 holds.
If the following inequality is satisfied

\begin{equation}
\left| G(j\omega)-C\right| > R \hspace{2mm}\forall \hspace{2mm} \omega
\, ,
\label{eq:distance criterion}
\end{equation}
then the Nyquist curve of $G(s)$ does not enter the circle  with a radius of $\frac{k_2-k_1}{2k_1k_2}$ and center at $\frac{-k_2-k_1}{2k_1k_2}$ and, thus, the wind speed estimator is globally convergent. 

In equation \eqref{eq:distance criterion}, it holds $C=\frac{-k_2-k_1}{2k_1k_2}$ and $R=\frac{k_2-k_1}{2k_1k_2}$;
$C$ and $R$ represent the center and the radius of the circle.
In this respect, \eqref{eq:distance criterion} states that the distance between $G(j\omega)$ and $C$ should be larger than $R$ for all $\omega$.
\end{corollary}

\begin{rem}
The scaling factor $\alpha$ is determined by placing the circle center $C$ to $(-1,0)$ in the Nyquist diagram.
Thus, $\alpha$ is given by
\begin{equation}
\alpha = \frac{k_2+k_1}{2k_1 k_2} 
\, .
\label{eq:alpha}
\end{equation}
\end{rem}

%\begin{rem}
%For the computation of the distance in~\eqref{eq:distance criterion}, the time-delay transfer function $e^{-sT}$ can be modelled with the well-known Padé approximation~\cite{Basdevant_1972}. 
%\end{rem}
Based on the circle criterion, the global convergence of the proposed estimator with a PI correction term is proved, which opens up the inclusion of time-delays and/or uncertainties.

\section{Case study}\label{sec:5}
The main objective of this section is twofold. 
First, we demonstrate that the circle criterion is an alternative method to analyze its global convergence. 
In the case study, different PI gains and time-delays are considered for discussions.
Second, the performance of the proposed wind speed estimator with a PI correction term is illustrated by means of a comparison to the original I\&I estimator.
%In this section, the proposed circle criterion wind speed estimator with a PI correction term is verified via several case studies.

The wind turbine dynamics is simulated using NREL's Fatigue, Aerodynamics, Structures, and Turbulence (FAST) tool~\cite{Jonkman-2005}.
As the wind speed estimator is employed under closed-loop interconnection with control systems, the classical $K\omega_g^2$ torque controller~\cite{Bossanyi_2000}, with $K$ being a predetermined optimal gain, is implemented.
% On the other hand, the step uniform wind speed is considered to evaluate the performance of the estimator in the case where the wind speed is not constant.
%The classical power-law wind profile model is used with a power--law exponent of $0.2$.
% The amplitude of the wind speed varies from 5m/s to 9m/s.
In order to evaluate the performance of the estimators under non-constant wind speed, stepwise wind from 5~m/s to 9~m/s with 2~m/s step size is considered.
%The wind shear and turbulence are not taken into consideration in the wind inflow condition.
Based on the model setup, the effectiveness of the wind speed estimator is demonstrated in the simulations.

First, the performance of the proposed wind speed estimator with a PI correction term and the original I\&I estimator is illustrated in Fig.~\ref{Model compare} for comparisons.
The proportional gain $\gamma$ and the delayed time $T$ are set to $80$ and $0.1$s, respectively, for both of them, while the integral gain $\beta$ is selected as $4$ for the proposed estimator.

\begin{figure}
\centering 
\includegraphics[width=0.85\columnwidth]{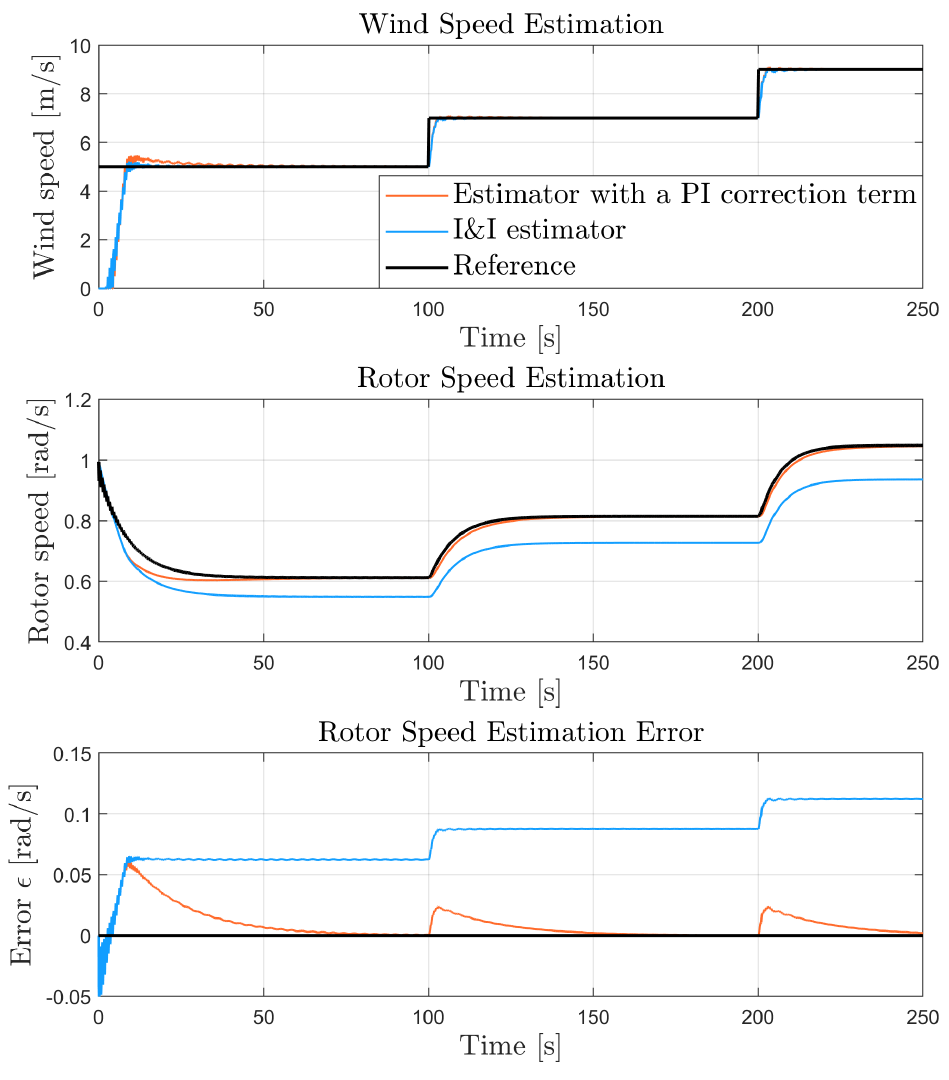}
\caption{Comparisons between the proposed wind speed estimator with a PI correction term and the original I\&I estimator.}
\label{Model compare} %
\end{figure}
In comparison, the proposed estimator with PI correction term in general shows similar results as the original I\&I estimator for the considered uniform wind speed conditions.
By adding the integrator into the correction term, it is clear that the steady-state error on the rotor speed estimation is reduced to zero, thus leading to an unbiased estimate of the rotor speed, \emph{i.e.}, $\hat{\omega}_r$.
Therefore, not only the wind speed, but also the rotor speed can be estimated by the proposed estimator with an additional integrator correction term.

Second, the effect of the proportional and integral gains on the global convergence of the estimator is investigated in the following three cases:
%Different values are selected for the gains to illustrate their effects on the wind speed estimation results. That is,
Case~1: $\gamma=40$, $\beta=10$, $T=0.1$s;
Case~2: $\gamma=100$, $\beta=10$, $T=0.1$s;
Case~3: $\gamma=100$, $\beta=200$, $T=0.1$s.
The simulation results are depicted in Fig.~\ref{gain_compare}. 
As anticipated, larger gain values speed up the convergence of the estimator, whereas it may induce the instability of the system.
The wind speed estimator shows higher overshoots at the transition of different wind speed steps in Case 3.
In particular, significant oscillations are observed in Case 3 at the wind speed of 5m/s.
%Their stability analyses are intuitively presented in the Nyquist diagram in Fig.~\ref{gain_nyquist}, where $k_1$ and $k_2$ are $0.016$ and $0.095$, respectively.
It is clear from Fig.~\ref{gain_nyquist} that the Nyquist curve of $G(s)$ in Cases 1-2 does not enter the circle, which indicates its global convergence ability.
However, for the third case the stability is not guaranteed as indicated by the circle criterion. 
The sufficient stability conditions can be derived according to the proposed distance criterion in~\eqref{eq:distance criterion}.
In the case where $\gamma=100$ and $T=0.1$s, $\beta<59.5$ can be selected to satisfy the distance criterion, and thus achieve the globally convergent wind speed estimation.

\begin{figure}
\centering 
\includegraphics[width=0.9\columnwidth]{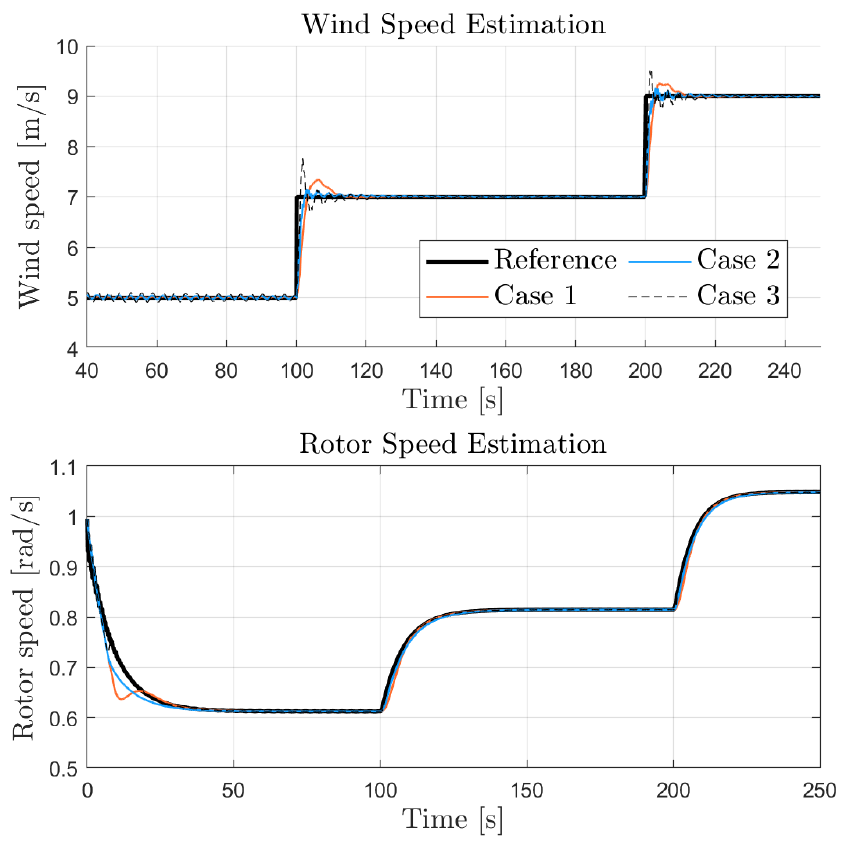}
\caption{Comparisons of the wind speed and rotor speed estimations between different gains under the stepwise wind condition.}
\label{gain_compare} %
\end{figure}

\begin{figure}
\centering 
\includegraphics[width=0.9\columnwidth]{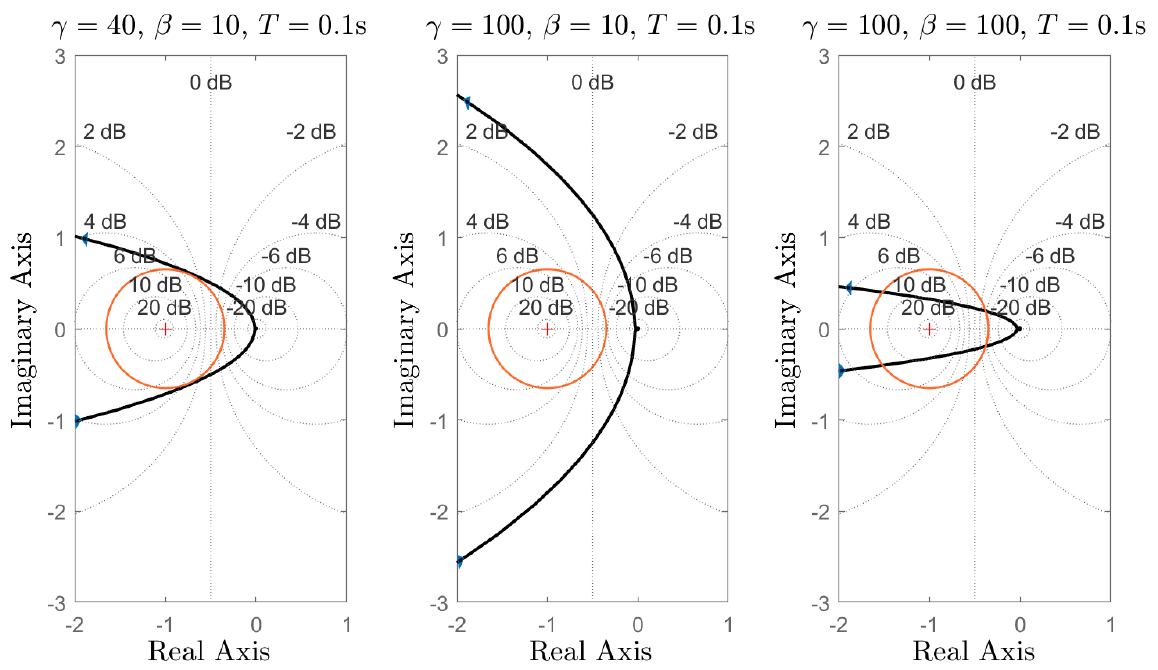}
\caption{Nyquist plot of $G(s)$ with different gains and the circle defined in the circle criterion.}
\label{gain_nyquist} %
\end{figure}

Third, the global convergence of the estimator for different time-delayed systems is analyzed as below:
%Different time-delays are considered in this scenario. 
%That is,
Case~4: $\gamma=40$, $\beta=10$, $T=0.1$s;
Case~5: $\gamma=40$, $\beta=10$, $T=0.6$s;
Case~6: $\gamma=40$, $\beta=10$, $T=2$s.
As seen in Fig.~\ref{delay_compare}, a larger time-delay, \emph{i.e.}, $2$s in Case 6, induces significant oscillations of the estimator.
The amplitude of the oscillations seems to grow unboundedly.
This is actually anticipated as the circle criterion presented in Fig.~\ref{delay_nyquist} implies that 
the estimator is unstable for such a time-delay.
%For the rest cases, the Nyquist curve of $G(s)$ enters the circle, and thus the estimator loses the asymptotically (absolute) stability.
To achieve the global convergence for Case 6, gains or time-delays should be reduced to satisfy the circle criterion design.
For instance, a lower time-delay, \emph{i.e.}, $T<0.174$s, can be selected for Case 6 to achieve its global convergence according to~\eqref{eq:distance criterion}.

\begin{figure}
\centering 
\includegraphics[width=0.9\columnwidth]{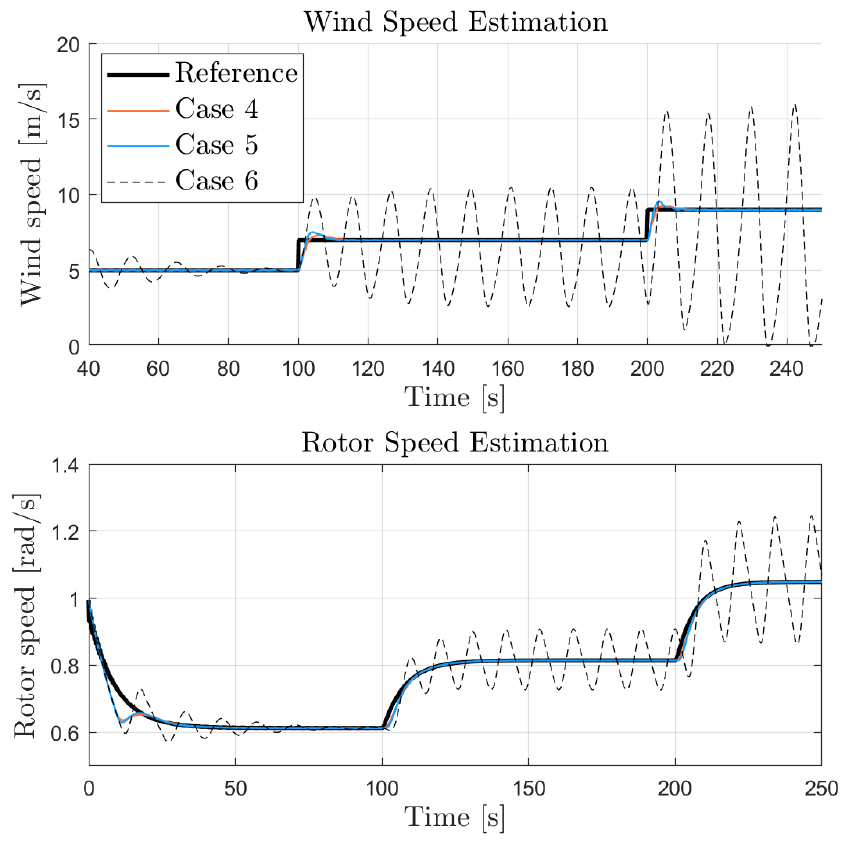}
\caption{Comparisons of the estimator with different time-delays under the stepwise wind condition. }
\label{delay_compare} %
\end{figure}

\begin{figure}
\centering 
\includegraphics[width=0.9\columnwidth]{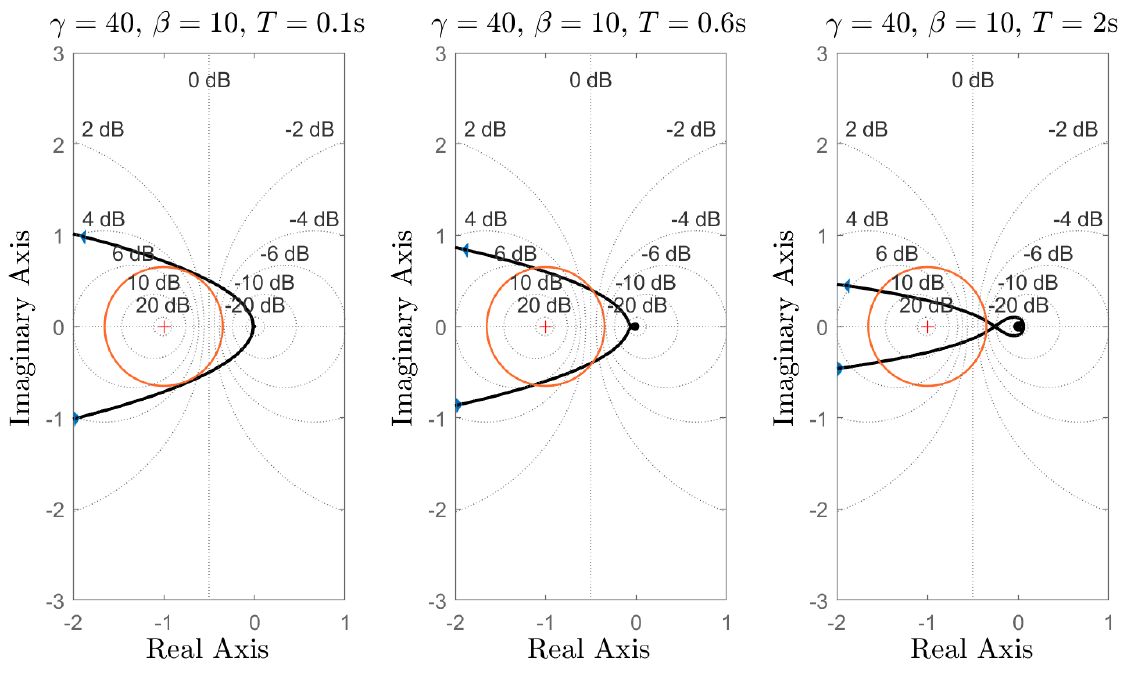}
\caption{Nyquist plot of $G(s)$ with different time-delays and the circle defined in the circle criterion.}
\label{delay_nyquist} %
\end{figure}

In conclusion, the circle criterion is an effective alternative method to achieve the global convergence of the estimator for time-delayed systems.
In case the sufficient stability condition is satisfied, the wind speed estimator with a PI correction term shows good performance on both wind speed and rotor speed estimations.

\vspace{-0.2cm}

\section{Conclusions}\label{sec:6}
In this letter, we illustrate that the circle criterion can be used as an alternative global convergence proof to the Immersion and Invariance (I\&I) estimator which also opens up the inclusion of time-delays and uncertainties. 
In detail, we show that the I\&I estimator can be rewritten as an estimator with a proportional correction term.
By looking at the theoretical framework from a new perspective, we propose to include an additional integrator to the correction term to improve the estimator performance.
The nonlinearity in the estimator appears as a power coefficient and can be sector bounded. This allows to use
the circle criterion to prove the global convergence of the estimator.
%Case studies are presented to demonstrate the effectiveness of the proposed estimator.
Case studies exhibit that the proposed wind speed estimator with a proportional integral correction term shows good performance on wind speed estimation. 
As the steady-state error on the rotor speed estimation is eliminated by the additional integrator correction term, the proposed estimator can also provide an accurate estimate of the rotor speed.
%The effects of the gains and time-delays on the global convergence of the estimator are investigated.
Furthermore, simulation results demonstrate that the circle criterion is an alternative method to prove its global convergence.
Based on this, the effects of the gains and time-delays on the global convergence of the estimator are analysed.

%As the sufficient stability condition is satisfied, the proposed method is able to achieve a globally convergent wind speed estimation. 

%As presented in this paper, there exists an offset in such a wind speed estimator which can not be cancelled by the proportional gain.  
%Current research is under way to address this issue and, hence, improve the estimation performance.
%Future work will include more complicated wind inflow conditions, such as wind shear, turbulence, wind farm wake, etc. 
%Furthermore, the blade effective wind speed estimator will be investigated.
%\vspace{-0.1cm}
\vspace{-0.2cm}

\bibliographystyle{IEEEtran} 
\bibliography{references}

\begin{thebibliography}{10}
\providecommand{\url}[1]{#1}
\csname url@samestyle\endcsname
\providecommand{\newblock}{\relax}
\providecommand{\bibinfo}[2]{#2}
\providecommand{\BIBentrySTDinterwordspacing}{\spaceskip=0pt\relax}
\providecommand{\BIBentryALTinterwordstretchfactor}{4}
\providecommand{\BIBentryALTinterwordspacing}{\spaceskip=\fontdimen2\font plus
\BIBentryALTinterwordstretchfactor\fontdimen3\font minus
  \fontdimen4\font\relax}
\providecommand{\BIBforeignlanguage}[2]{{%
\expandafter\ifx\csname l@#1\endcsname\relax
\typeout{** WARNING: IEEEtran.bst: No hyphenation pattern has been}%
\typeout{** loaded for the language `#1'. Using the pattern for}%
\typeout{** the default language instead.}%
\else
\language=\csname l@#1\endcsname
\fi
#2}}
\providecommand{\BIBdecl}{\relax}
\BIBdecl

\bibitem{GWEC_2020}
J.~Lee and F.~Zhao, ``Global wind statistics 2020,'' Global wind energy
  council, Report, 2020.

\bibitem{Ostergaard_2007}
K.~Z. {\O}stergaard, P.~Brath, and J.~Stoustrup, ``Estimation of effective wind
  speed,'' \emph{Journal of Physics: Conference Series}, vol.~75, p. 012082,
  jul 2007.

\bibitem{DOEKEMEIJER2020719}
B.~M. Doekemeijer, D.~{van der Hoek}, and J.~W. {van Wingerden}, ``Closed-loop
  model-based wind farm control using {FLORIS} under time-varying inflow
  conditions,'' \emph{Renewable Energy}, vol. 156, pp. 719--730, 2020.

\bibitem{Liu2021}
Y.~{Liu}, A.~{Kusumo Pamososuryo}, R.~{Ferrari}, T.~{Gybel Hovgaard}, and J.~W.
  {van Wingerden}, ``{Blade Effective Wind Speed Estimation: A Subspace
  Predictive Repetitive Estimator Approach},'' in \emph{2021 European Control
  Conference (ECC)}, 2021.

\bibitem{Soltani2013}
M.~N. Soltani, T.~Knudsen, M.~Svenstrup, R.~Wisniewski, P.~Brath, R.~Ortega,
  and K.~E. Johnson, ``{Estimation of rotor effective wind speed: A
  comparison},'' \emph{IEEE Transactions on Control Systems Technology},
  vol.~21, no.~4, pp. 1155--1167, 2013.

\bibitem{Ma_1995}
X.~Ma, N.~Poulsen, and H.~Bindner,
  \emph{\BIBforeignlanguage{English}{Estimation of Wind Speed in Connection to
  a Wind Turbine}}.\hskip 1em plus 0.5em minus 0.4em\relax Informatics and
  Mathematical Modelling, Technical University of Denmark, DTU, 1995.

\bibitem{Ortega_2013}
R.~Ortega, F.~Mancilla-David, and F.~Jaramillo, ``A globally convergent wind
  speed estimator for wind turbine systems,'' \emph{International Journal of
  Adaptive Control and Signal Processing}, vol.~27, no.~5, pp. 413--425, 2013.

\bibitem{Khalil_2015}
H.~Khalil, \emph{Nonlinear Control, Global Edition}.\hskip 1em plus 0.5em minus
  0.4em\relax Pearson Education Limited, 2015.

\bibitem{Lee_1981}
T.~{Lee} and S.~{Dianat}, ``Stability of time-delay systems,'' \emph{IEEE
  Transactions on Automatic Control}, vol.~26, no.~4, pp. 951--953, 1981.

\bibitem{Jonkman_2009}
J.~Jonkman, S.~Butterfield, W.~Musial, and G.~Scott, ``Definition of a 5{MW}
  reference wind turbine for offshore system development,'' National Renewable
  Energy Laboratory, NREL/TP-500-38060, 2009.

\bibitem{Fernando_2012}
F.~Mancilla-David and R.~Ortega, ``Adaptive passivity-based control for maximum
  power extraction of stand-alone windmill systems,'' \emph{Control Engineering
  Practice}, vol.~20, no.~2, pp. 173--181, 2012.

\bibitem{Karl_2020}
K.~J. Åström and R.~M. Murray, \emph{\BIBforeignlanguage{English}{Feedback
  Systems: An Introduction for Scientists and Engineers}}.\hskip 1em plus 0.5em
  minus 0.4em\relax Princeton University Press, 2020.

\bibitem{Jonkman-2005}
J.~M. Jonkman and M.~L. Buhl, ``{FAST} user's guide,'' \emph{SciTech Connect:
  FAST User's Guide}, 2005.

\bibitem{Bossanyi_2000}
E.~A. Bossanyi, ``The design of closed loop controllers for wind turbines,''
  \emph{Wind Energy}, vol.~3, no.~3, pp. 149--163, 2000.

\end{thebibliography}

\end{document}